\documentclass[journal,onecolumn]{IEEEtran}
\usepackage{amsmath,graphicx,amsfonts,amssymb,latexsym}
\usepackage{amssymb}
\usepackage{color}
\usepackage{pgfplots}
\usepackage{subfigure}
\usepackage{algorithm}
\usepackage{graphicx}
\usepackage{fullpage}
\usepackage{subfigure}
\usepackage{epsfig}
\usepackage{epstopdf}
\newtheorem{proposition}{Proposition}
\newtheorem{theorem}{Theorem}
\newtheorem{proof}{Proof}

\begin{document}

\title{Filtering from Observations on Stiefel Manifolds}

\author{J\'er\'emie~Boulanger, Salem~Said, Nicolas~Le~Bihan,
and~Jonathan~H.~Manton,~\IEEEmembership{Senior~member}
\IEEEcompsocitemizethanks{
\IEEEcompsocthanksitem J. Boulanger is with the University of Rouen, France. email: {\tt jeremie.boulanger1@univ-rouen.fr}.\protect\\
\IEEEcompsocthanksitem S. Said is with the IMS, University of Bordeaux, France. email: {\tt salem.said@ims-bordeaux.fr}
\IEEEcompsocthanksitem N. Le Bihan is with the CNRS, University of Melbourne, Australia. email: {\tt nicolas.le-bihan@gipsa-lab.grenoble-inp.fr}. His research was supported by the ERA, European Union, through the International Outgoing Fellowship (IOF GeoSToSip 326176) program of the 7th PCRD.
\protect\\
\IEEEcompsocthanksitem J.H. Manton is with the University of Melbourne, Australia. email: {\tt jmanton@unimelb.edu.au} }
}


\maketitle

\begin{abstract}
This paper considers the problem of optimal filtering for partially observed signals taking values on the rotation group. More precisely, one or more components are considered not to be available in the measurement of the attitude of a 3D rigid body. In such cases, the observed signal takes its values on a Stiefel manifold. It is demonstrated how to filter the observed signal through the \emph{anti-development} built from observations. A particle filter implementation is proposed to perform the estimation of the signal partially observed and corrupted by noise. The sampling issue is also addressed and interpolation methods are introduced. Illustration of the proposed technique on synthetic data demonstrates the ability of the approach to estimate the angular velocity of a partially observed 3D system partially observed.
\end{abstract}


\section{Introduction}

In numerous engineering problems, systems with states having values and evolving on the special orthogonal group $SO(n)$ can be encountered \cite{Bloch,Lovera,Zamani11,barrau13,zamani13,Landis14}. In order to control such systems, their angular velocity must be estimated from possibly noisy measurements. This paper considers the case where only partial observations of the system are available, {\em i.e.} not all the components of the movement are recorded. The observation signal is modeled as a process taking its values on a Stiefel manifold. In addition, the presence of a multiplicative noise is considered in the observation process. Classical methods, including extended Kalman filter \cite{Bourmaud,Canet} can not be applied directly here as they rely on the independent increments assumption. As explained later, it is not the case in the model we consider here. We propose to use the {\em anti-development signal} computed from the observed data. We present the way to build this signal, and adress the sampling/interpolation issue as weel. We also demonstrate how to perform optimal filtering on the anti-development signal. A numerical solution (particle filter) via a Monte-Carlo method is provided to perform this filtering and illustrated on the Stiefel manifold ${\cal S}^2$. The proposed technique is however valid for higher dimension Stiefel manifolds.

The rest of the paper is organized as follows. Section \ref{sec_geom} presents the geometry of Stiefel manifolds based on the geometry of $SO(n)$ and the concept of horizontal space. Section \ref{sec_filt_st} presents a time continuous theoretical solution to the filtering problem with observations in Stiefel manifolds. As opposed to the usual case, the noise cannot be considered additive anymore here in our model. The proposed solution is based on the antidevelopment, a defined with respect to the observation process that satisfies an additive noise model. Section \ref{sec_interp} presents a theorem to overcome the problem of discrete sampling. Section \ref{sec_impl} gives a practical solution based on a Monte-Carlo method for filtering. Finally Section \ref{sec_son} considers the case of observation in $SO(n)$ and compare different approximation to the optimal solution.

\section{Geometry of Stiefel manifolds}
\label{sec_geom}

The Stiefel manifold $V_{n,k}$ 
is the set of orthonormal $k$-frames in $\mathbb{R}^n$. It is well known and used in linear algebra to describe principal subspaces \cite{Edelman} and has found applications in sensors array \cite{Micka}, statistics \cite{chikuze}, optimization \cite{absil}, channel estimation in wireless communications \cite{Ozdemir} or in light independent scene representation in computer vision \cite{Lui09}.

First, recall that a $n \times n$ matrix $R$ with real components is an element of the rotation group $SO(n)$ if it is orthogonal and has a unit determinant. This is to say that $R \in SO(n)$ iff:
\begin{align}
R^TR=I_n \quad \text{and} \quad \det R=1
\end{align}
where $I_n$ denotes the $n\times n$ identity matrix. Intuitively, $SO(n)$ is the set of positively oriented orthonormal basis vectors of $\mathbb{R}^n$.

In ${\mathbb R}^n$, the Stiefel manifold $V_{n,k}$ is defined as the set of matrices $P\in\mathbb{R}^{n\times k}$ such that:
\begin{equation}
P^TP=I_k
\label{eq_def_stiefel}
\end{equation}
and with $k \leq n$. For example, if $k=1$, then $V_{n,1}$ is the hypersphere $S^{n-1}$, {\em i.e.} the set of unit vectors in $\mathbb{R}^n$. If $k=n$, then $V_{n,n}$ corresponds to the orthogonal group $O(n)$.

Let $\Pi:SO(n)\rightarrow V_{n,k}$ be the projection consisting in the truncation of the $n-k$ last columns of a rotation matrix, and let us denote:
\begin{equation}
\Pi(R)=P
\label{eq_proj}
\end{equation}

When $k\leq n-2$, the projection is not injective. In this case, a matrix $P\in V_{n,k}$ can be completed by different sets of orthonormal vectors to form an oriented orthonormal basis of $\mathbb{R}^n$, which means that in such cases:
$$
\Pi(R_1)=\Pi(R_2) \Leftrightarrow R_1=R_2\left(\begin{matrix}I_k&0\\0&C\end{matrix}\right)
$$
with $R_1,R_2\in SO(n),\  C\in SO(n-k)$.

However, if $k\leq n$, then $\Pi$ is clearly surjective, \textit{i.e} $\Pi(SO(n))\in V_{n,k}$ as the $k$ first columns of a rotation matrix are orthonormal vectors. 
Therefore, $\Pi(R)^T\Pi(R)=I_k$ for $R\in SO(n)$. The Stiefel manifold $V_{n,k}$ can then be described as:
\begin{equation}
V_{n,k}=\left\{\Pi(R),R\in SO(n)\right\}
\label{eq_def_stiefel2}
\end{equation}

Note that the case $k=n$ needs special care. Indeed, $V_{n,n}=O(n)$ is the group of orthonormal matrix and is composed of two connected components: the set of orthonormal matrices with a positive determinant (positively oriented basis) $SO(n)$ and the set of orthonormal matrices with a negative determinant (negatively oriented basis). 
In this study, we will consider continuous random processes $P_t \in V_{n,k}$ which will solely belong to the same component of $V_{n,k}$ as their initial value $P_0$ belongs to. Therefore, if $\det(P_0)=+1$, then $\Pi=Id$ covers all the reachable points in the Stiefel manifold from $SO(n)$. If $\det(P_0)=-1$, considering $\Pi(R)$ as the application reversing the sign of the last column of $R$ allows $\Pi$ to cover all the reachable points in the Stiefel manifold from $SO(n)$. Consequently, expression  (\ref{eq_def_stiefel2}) can be extended to the case where $k=n$, by considering only one connected component. This case will be considered in Section \ref{sec_son}.

As $V_{n,k}$ can be constructed from $SO(n)$, we now investigate how the geometry of $V_{n,k}$ can be described using the geometry of $SO(n)$. From its definition, the projection $\Pi$ is left invariant:
\begin{equation}
R_1\Pi(R_2)=\Pi(R_1R_2)
\label{eq_action_group}
\end{equation}
with $R_1,R_2\in SO(n)$.

As $\Pi$ is surjective, one also get the action of $SO(n)$ on $V_{n,k}$. If $R\in SO(n)$ and $P\in V_{n,k}$, then $RP\in V_{n,k}$. This property will be used later on to describe a process on $V_{n,k}$ via the action of $SO(n)$. This group action can be visualized by considering the example of the sphere $V_{3,1} \cong {\cal S}^2$. Points at the surface of the sphere ${\cal S}^2$ can reach all the locations on this manifold through the transitive action of $SO(3)$ on the sphere: $SO(3) \times {\cal S}^2 \rightarrow {\cal S}^2$.

First, let us identify the tangent bundle of $V_{n,k}$, denoted $TV_{n,k}$. It will be of use in Section \ref{sec_filt_st} to define stochastic processes via the action of $SO(n)$ in the space tangent to a point in $V_{n,k}$. Denote $\mathfrak{so}(n)$ the Lie algebra\footnote{$\mathfrak{so}(n)$ is the algebra of real-valued anti-symmetric matrices of size $n \times n$.} associated to the Lie group $SO(n)$ and let $\chi:\mathfrak{so}(n)\times V_{n,k}\rightarrow TV_{n,k}$ be the application defined by:
\begin{equation}
\chi(\sigma,P)=\left(\left.\frac{d}{dt}\exp(t\sigma)P\right)\right|_{t=0} =\sigma P
\end{equation}

We can show by inclusion and dimension equality that $\chi(.,P)$ is surjective onto $T_PV_{n,k}$, \textit{i.e} $T_PV_{n,k}=\{\sigma P, \sigma\in \mathfrak{so}(n)\}$, where we used the notation $T_PV_{n,k}$ for the tangent space attached to a point $P \in V_{n,k}$.

Now, for a given point $P \in V_{n,k}$, let $R\in SO(n)$ be a pre-image of $P$ via $\Pi$, \textit{i.e} $P=\Pi(R)$. As $\Pi$ is surjective, $\Pi^{-1}(P)\neq\varnothing$ and $R$ is well defined. Then, the vertical space \cite{Lui09} at the point $P$, denoted $\mathcal{V}_R$, is defined as:
\begin{equation}
\mathcal{V}_R=\operatorname{Ker}d\Pi_R
\label{equation}
\end{equation}
where $d\Pi_R$ is the differential of $\Pi$ at the point $R$. By definition of $d\Pi_R$, the vertical space $\mathcal{V}_R$ is a subspace of the tangent space $T_RSO(n)$. In the case when $n=3$ and $k=1$, then $V_{3,1} \cong S^2$ and $P$ is a point on the unit sphere in ${\mathbb R}^3$. The vertical space corresponds to the set of rotations which have their axis aligned with $P$. Such rotations leave $P$ invariant. Figure \ref{fig_vertic_v2} displays a graphical interpretation of the vertical space $\mathcal{V}_R$.

\begin{figure}[h!]
\centering
\subfigure{\epsfig{figure=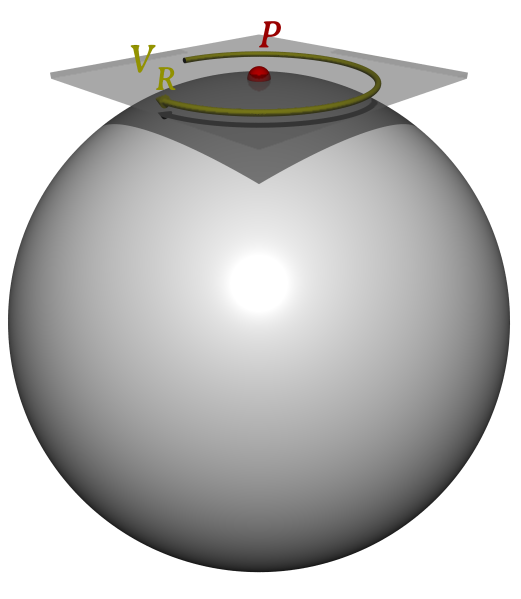,width=5cm,height=5.5cm}}
\caption{Graphical representation of the vertical space ${\cal V}_R$ at $P$ for the case $P \in V_{3,1}$. For a rotation $R$ acting on $P$,  ${\cal V}_R$ is the orthogonal complement of ${\cal H}_R$ in $T_PV_{3,1}$. If the axis of the rotation of $R$ is parallel to $P$, then its action is in ${\cal V}_R$ and it is not visible as $P$ is rotating about itself. 
\label{fig_vertic_v2}}
\end{figure}

\begin{figure}[h!]
\centering
\subfigure{\epsfig{figure=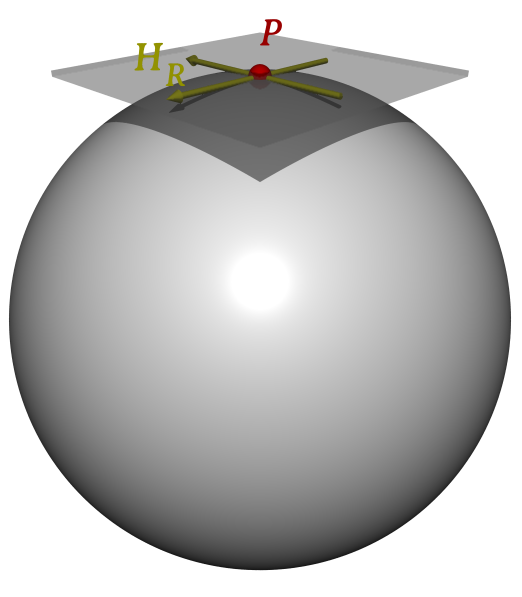,width=5cm,height=5.5cm}}
\caption{Graphical representation of the horizontal space ${\cal H}_R$ at $P$ for the case $P \in V_{3,1}$. For a rotation $R$ acting on $P$, ${\cal H}_R$ is a subspace of $T_PV_{3,1}$. If the axis of rotation of $R$ is orthogonal to $P$, then its action is in $\mathcal{H}_R$ and it is visible.
\label{fig_horiz_s2}}
\end{figure}

Making use of the standard scalar product on $\mathfrak{so}(n)$ which reads for any skew-symmetric matrices $\sigma,\varsigma \in \mathfrak{so}(n)$ like: 
\begin{equation}
<\sigma,\varsigma>=\frac{1}{2}\operatorname{tr}(\sigma^T\varsigma)
\end{equation}
it is possible to construct the orthogonal complement of $\mathcal{V}_R$ in $T_RSO(n)$, called the {\em  horizontal space} and denoted $\mathcal{H}_R$. We have then that:
\begin{equation}
T_RSO(n)=\mathcal{V}_R\oplus \mathcal{H}_R
\end{equation}

Figure \ref{fig_horiz_s2} displays a graphical interpretation of the horizontal space $\mathcal{H}_R$.

As $d\Pi_R$ is linear, the restriction of $d\Pi_R$ to $\mathcal{H}_R$ is bijective. In other words, $T_PV_{n,k}$ and $\mathcal{H}_R$ are isomorphic. For a vector $v\in T_PV_{n,k}$, let $v^{\mathcal{H}}\in T_RSO(n)$ be the vector defined as:
\begin{equation}
d\Pi_R(v^{\mathcal{H}})=v
\end{equation}


For example, consider again the case of the Stiefel manifold $V_{3,1} \cong {\cal S}^2$. Considering $P=(1,0,0)^T\in V_{3,1}$, the matrices $R_1=(e_1,e_2,e_3)$ and $R_2=(e_1,e_3,-e_2)$ for $\{e_i\}_{i\leq3}$ the canonical basis of $\mathbb{R}^3$ are both pre-images of $P$, \emph{i.e.} $\Pi(R_1)=\Pi(R_2)=P$. The application $\chi(.,P)$ describes the tangent space $T_PV_{3,1}$ like $\chi(\sigma,P)=(0,\sigma_{21},\sigma_{31})^T$ where $\sigma_{ij}$ is the $(i,j)$ matrix elements of $\sigma \in \mathfrak{so}(3)$. At the pre-image $R_1$, the vertical and horizontal spaces are thus defined as:
$$
\mathcal{V}_{R_1}=\left\{\left(\begin{matrix}0&0&0\\0&0&\alpha\\0&-\alpha&0\end{matrix}\right),\alpha\in\mathbb{R}\right\}
$$
and:
$$
\mathcal{H}_{R_1}=\left\{\left(\begin{matrix}0&\beta&\gamma\\-\beta&0&0\\-\gamma&0&0\end{matrix}\right),\beta, \gamma\in\mathbb{R}\right\}.
$$
At the pre-image $R_2$, the vertical and horizontal spaces are defined as:
$$
\mathcal{V}_{R_2}=\left\{\left(\begin{matrix}0&0&0\\0&\alpha&0\\0&0&\alpha\end{matrix}\right),\alpha\in\mathbb{R}\right\}
$$
and:
$$
\mathcal{H}_{R_2}=\left\{\left(\begin{matrix}0&\beta&\gamma\\\gamma&0&0\\-\beta&0&0\end{matrix}\right),\beta, \gamma\in\mathbb{R}\right\}.
$$
One can direclty check that spaces $\mathcal{H}_{R_1}R_1^T$ and $\mathcal{H}_{R_2}R_2^T$ are identical. This is true even if the horizontal subspaces $\mathcal{H}_{R_1}$ and $\mathcal{H}_{R_2}$ are different, and  is a consequence of the fact that they are defined by a different pre-image of $P$. A graphical illustration of the notion of horizontal and vertical spaces is displayed in figure \ref{fig_horiz}

\begin{figure}
\centering
\centerline{\epsfig{figure=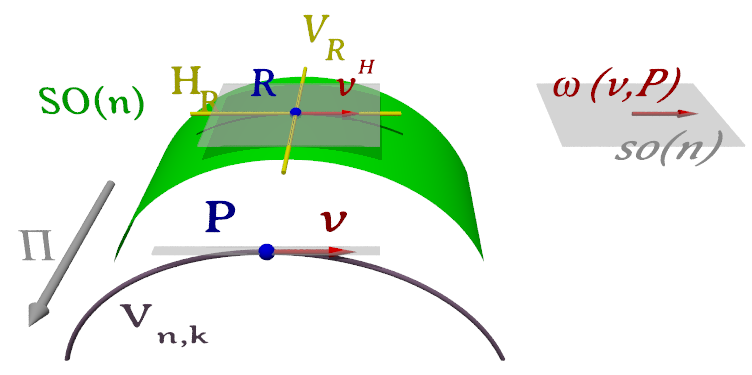,width=8cm,height=4cm}}
\caption{Illustration of the different notions introduced to describe the Stiefel manifold $V_{n,k}$ as the image from the projection $\Pi$ of $SO(n)$. The horizontal $\mathcal{H}_R$ and vertical $\mathcal{V}_R$ spaces are dependent of the chosen pre-image but the translation into $\mathfrak{so}(n)$ via $\omega$ is invariant with respect to the choice of the pre-image.
\label{fig_horiz}}
\end{figure}

Due to the isomorphism between $T_PV_{n,k}$ and $\mathcal{H}_R$, the application $\chi(.,P)$ restricted to $\mathcal{H}_RR^T\in\mathfrak{so}(n)$ with $\Pi(R)=P$ is bijective. In other words, $\chi(.,P)|_{\mathcal{H}_R}$ is invertible. Let $\omega:TV_{n,k}\rightarrow\mathfrak{so}(n)$ denote this inverse and let us call it the \emph{restricited inverse}. It then reads:
\begin{equation}
\omega(v,P)=v^\mathcal{H}R^T.
\label{eq_x_inv}
\end{equation}
The term $v^{\mathcal{H}}$ is the horizontal vector from the tangent space to $R$. Despite the definition of $\omega$ being dependent on $R$, this is not the case because $v^{\mathcal{H}}$ also depends on $R$, and, in the end, the term $v^\mathcal{H}R^T$ ii independent of $R$. Finally, it is possible to define a metric on $V_{n,k}$ using the metric on $SO(n)$. Let $<.,.>_P$ be the metric defined as:
\begin{equation}
<v_1,v_2>_P=<v_1^{\mathcal{H}},v_2^{\mathcal{H}}>_R
\label{eq_inner_st}
\end{equation}
for any two vectors $v_1,v_2\in T_PV_{n,k}$, with $\Pi(R)=P$, and where $<.,>_R$ denotes the scalar product defined in $T_RSO(n)$.

\section{Filtering from observations on Stiefel manifolds}
\label{sec_filt_st}


We consider the problem of a partially observed system whose state evolves on the rotation group $SO(n)$. In practice, such observations may come from flawed sensors or devices, leading to the availability of a limited part of the signal to filter. For example, in the context of satellite's control, existing algorithms require the knowledge of the angular velocity and the orientation of the satellite to monitor its orientation \cite{Lovera}\cite{Bloch}. This angular velocity is determined from different internal sensors. However, if some of these sensors become faulty, the velocity of the satellite is no more available and the satellite cannot be controlled properly anymore. 

The presented algorithm proposes to tackle the problem of lack in parts of the signal to filter and takes advantage of the available observations to perform optimal filtering. More precisely, we present a technique to obtain an estimate of the velocity of the system with only partial observations of its orientation, \emph{i.e.} partial observations on $SO(n)$.

\subsection{Observation model}

The model considered is as follows: a process $S_t\in SO(n)$ is defined by its angular velocity $x_t\in \mathfrak{so}(n)$ where $t$ represents time. Our aim is to obtain an estimate of the angular velocity $x_t$ based on observations of $S_t$ which are not complete as well as noisy. The process $x_t$ is here assumed to be the solution of the following linear stochastic differential equation in $\mathfrak{so}(n)$:
\begin{equation}
\label{eq_model}
dx_t=Fx_t+db_t
\end{equation}
where $b_t\in\mathfrak{so}(n)$ is a Brownian motion with variance $\sigma_b^2$. In this case, $x_t$ is a Markov process and its transition kernel for time $t+s$ based on $x_s$ is denoted $q_t(x_s,\ .)$.

The partial observation is here modeled as a process $P_t$ on the Stiefel manifold $V_{n,k}$, \textit{i.e} only $k$ components of $S_t$ amongst the total of $n$ components are known. The filtering problem then reads: we want to estimate $x_t\in \mathfrak{so}(n)$ from $P_t\in V_{n,k}$ defined as $P_t=\Pi(S_t)$ in the presence of noise. The noise is modeled by a Brownian motion $w_t\in\mathfrak{so}(n)$ with variance $\sigma_w^2$ independent from $x_t$ acting in $\mathfrak{so}(n)$. As $x_t$ is the angular velocity of the observed processed $P_t$, then $P_t$ is solution of the stochastic differential equation:
\begin{equation}
dP_t=\left(x_tdt+\circ dw_t\right)P_t
\label{eq_obs_stie}
\end{equation}
where notation $\circ$ is used to denote the Stratonovich integral.

Due to the presence of the noise $w_t$, $x_t$ cannot be exactly determined. Instead, we want to determine the distribution $\pi_t$ of $x_t$ conditioned by the observation of $\mathcal{P}_s=\{P_s,s\leq t\}$. It is possible to construct some estimator for $x_t$ based on its conditional distribution $\pi_t$.

It is noticeable that despite that the noise acts additively in the tangent space $T_{P_t}V_{n,k}$, it acts as a multiplicative noise for the process $P_t$, preventing us from using usual filtering methods. Indeed, classical methods like Kalman filter rely on the independence of the increments $dP_t$. However, this is not applicable in our case as the increments depends of $P_t$.

\subsection{The anti-development solution}

We propose a solution based on the concept of anti-development. It consists in constructing a process $z_t$ in one-to-one correspondence with $P_t$ such that $z_t$ is solution of a stochastic differential equation with additive noise. The likelihood used to compute the solution is then based on $z_t$.

Let $z_t\in\mathfrak{so}(n)$ and $R_t\in SO(n)$ be defined as:
\begin{equation}
\begin{aligned}
dz_t&=\omega(\circ dP_t,P_t)\\
dR_t&=(\circ dz_t)R_t
\end{aligned}
\label{eq_antidev_st1}
\end{equation}
where $\omega$ is as defined previously in (\ref{eq_x_inv}), the \emph{restricted inverse} of $\chi$.

The process $z_t$ is called the anti-development of $P_t$ and $R_t$ is called the horizontal lift of $P_t$ \cite{Abs08}. An illustrative example of the anti-development on $V_{3,1}=S^2$ is presented in figure \ref{fig_antidev_st}.
\begin{figure}
\center
\centerline{\includegraphics[width=.5\linewidth]{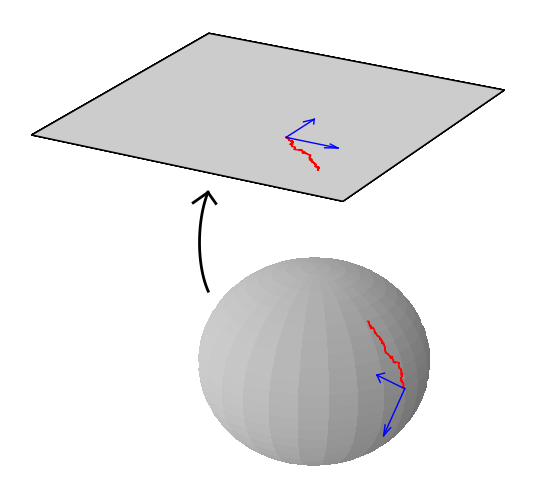}}
\centerline{\includegraphics[width=.5\linewidth]{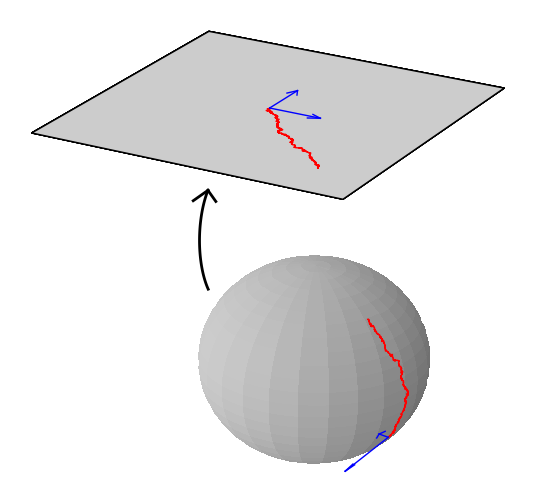}}
\centerline{\includegraphics[width=.5\linewidth]{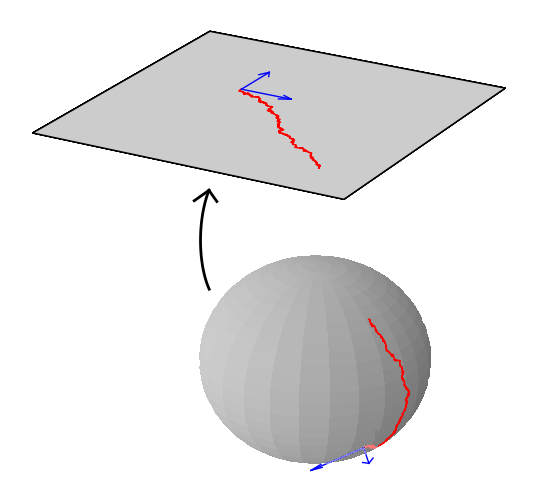}}
\caption{Example of a trajectory of $P_t$ (red on the sphere) on $V_{3,1}=S^2$ at three successive times (Top to Bottom). The anti-development $z_t$ is displayed in red in the plane over the sphere. It can be obtained by considering the trace left by the sphere when rolling without slipping on the plane, and with $P_t$ as a contact point. The anti-development is solution of a stochastic differential equation with additive noise, as opposed to $P_t$. \label{fig_antidev_st}}
\end{figure}

The process $z_t$ is the accumulation of the increments in the tangent space $T_PV_{n,k}$ whereas $R_t$ is the rotational process constructed by considering that the component in the vertical space is null. These processes are equivalent, in terms of information to $P_t$ because $P_t$ can be constructed like $P_t=R_tR^T_0P_0$ or $dP_t=\chi(dz_t,P_t)$.
Conditioning the distribution of $x_t$ by the observation of $P_t$ is equivalent as conditioning by the observation of $z_t$. However, the anti-development is a solution, as opposed to $P_t$, of a stochastic differential equation with additive noise.  It should be noticed that in general cases, $S_t\neq R_t$. Despite the vertical component has no action in the Stiefel manifold $V_{n,k}$, it still has some effect in $SO(n)$. This involves that in general, $S_t\neq R_t$. 

Replacing $dP_t$ in (\ref{eq_antidev_st1}) by its expression from (\ref{eq_obs_stie}) gives:
\begin{align}
dz_t&=\omega(dP_t,P_t)\notag\\
&=\omega\left(\left(x_tdt+\circ dw_t\right)P_t,P_t\right)\notag\\
&=\omega(\left(\chi\left(x_tdt+\circ dw_t\right),P_t\right),P_t)\notag\\
&=\omega(H_t,P_t)+\circ d\beta_t
\label{eq_antidev_st2}
\end{align}
where $H_t=\chi(x_t,P_t)$ and $d\beta_t=\omega\left(\chi\left(\circ dw_t,P_t\right),P_t\right)$.
By definition of $d\beta_t$, the process $\beta_t$ is constructed from the $k$ first components of $w_t$. Therefore, $\beta_t$ is a Brownian process with a variance that can be diagonalized as $\sigma_w^2I_{n,k}$ where $I_{n,k}=\operatorname{diag}(1,...,1,0,...,0)$ with $k$ non-zero elements.

This way, our filtering problem from observations in the Stiefel manifold $V_{n,k}$ with multiplicative noise is now reduced to a filtering problem in $\mathfrak{so}(n)$ with additive noise. For a test function $\phi$, we want to determine $\pi(\phi)=\mathbb{E}[\phi(x)|\mathcal{P}_t]$. The solution is therefore given by applying usual filtering methods \cite{Jazw} to the anti-development $z_t$ defined in (\ref{eq_antidev_st1}):
\begin{equation}
\pi_t(\phi)=\frac{\rho_t(\phi)}{\rho_t(1)},
\label{eq_ks_st}
\end{equation}
with $\rho_t(\phi)=\mathbb{E}\left[\phi(x')L_t(P,x')|\mathcal{P}_t\right]$, where $x'_t$ is a copy of $x_t$ independent of $P_t$ and $\mathcal{P}_t=\{P_s,s\leq t\}$. The likelihood $L_t$ is defined as:
\begin{equation}
L_t(P,x')=\exp\left( \frac{1}{\sigma_w^2}\int_0^t <x'_s,dz_s>-\frac{1}{2}||x'_s||^2 ds\right)
\label{eq_likelihood_sti}
\end{equation}
with $dz_t=\omega(\circ dP_t,P_t)$. By definition of the inner product in (\ref{eq_inner_st}), the likelihood in (\ref{eq_likelihood_sti}) can be rewritten like:
\begin{equation}
L_t(P,x')=\exp\left(\frac{1}{\sigma_w^2} \int_0^t <H'_s,dP_s>-\frac{1}{2}||H'_s||^2 ds\right)
\label{eq_likelihood_st}
\end{equation}
with $H'_s=\chi(x'_s,P_s)$ a copy of $H_s$ in the distribution sense. Expression (\ref{eq_likelihood_st}) is more amenable than the one from equation (\ref{eq_likelihood_sti}) as it does not require the computation of $R_t$. The integrand can directly be determined from the observations without constructing any auxiliary process. However, using expression (\ref{eq_likelihood_st}), the model for $H_t$ is not linear, even if $x_t$ is the solution of a linear model. As a consequence, expressing $\rho_t(\phi)$ is a complicated task. Nevertheless, it is still possible to get an approximation of the solution, using numerical methods, like the particle filter for example.
Before proposing a filtering solution, we address the issue due to the discrete nature of the observation of $P_t$.

\section{Interpolation function}
\label{sec_i1nterp}
It must be noted that the likelihood function $L_t$ given in (\ref{eq_likelihood_st}) requires the full observation of the process $\{P_s\}_{s\leq t}$ to compute the integrand. In practice, it is not possible to have a continuous observation of $P_t$. Only discrete samples are available. Let $\delta t$ be the sampling period. Between two samples, $P_t$ must be approximated using an interpolation function. This interpolation function must be chosen to minimize the approximation error as a function of the sampling period.

Let $\operatorname{Int}:V_{n,k}\times V_{n,k}\rightarrow {\mathfrak so}(n)$ be an interpolation function. It is thus required that $\operatorname{Int}$ should be such that given $\delta z_k=\operatorname{Int}(P_{k\delta t},P_{(k+1)\delta t})$, the likelihood based on discrete observation will converge to the continuous solution for $\delta t \rightarrow 0$.

\begin{theorem}
\label{prop1}
The Riemann sum
$$
\tilde{\mathcal{S}}_{n\delta t}=\sum_{k=0}^n <x_{k\delta t},\operatorname{Int}(P_{k\delta t},P_{(k+1)\delta t})>	
$$ 
with $n=t/\delta t$ converges towards 
$$
{\mathcal S}_t=\int_0^t <x_s,dz_s>
$$
in the sense $\mathbb{E}[|\tilde{\mathcal{S}}_{n\delta t}-\mathcal{S}_t|^2]\rightarrow 0$ when $\delta \rightarrow 0$ if the interpolation function $\operatorname{Int}:V_{n,k}\times V_{n,k}\rightarrow \mathfrak{so}(n)$ satisfies the following conditions:
\begin{itemize}
\item Its diagonal elements are nul, \emph{i.e.} $\operatorname{Int}(P,P)=0$ for all $P\in V_{n,k}$.
\item The function $\operatorname{Int}$ is $\mathcal{C}^2 \left(V_{n,k}\right)$.
\item $\nabla \operatorname{Int}(P,P)[v]=v^{\mathcal H} R^T$ for all $v \in T_P V_{n,k}$ and $\Pi(R)=P$.
\item $\nabla^2 \operatorname{Int}(P,P)[v]=0$ for all $v \in T_P V_{n,k}$.
\end{itemize}
Where the differentials  ($\nabla I$ and $\nabla^2 I$) are computed with respect to the second variable.
\end{theorem}

\begin{proof}
Considering the function $f=\operatorname{Int}(P_{k\delta t},.)$, then condition $2)$ allows the use of It\=o lemma:
\begin{align*}
&f(P_{(k+1)\delta t})-f(P_{k\delta t})=\int_{k\delta t}^{(k+1)\delta t}\nabla f[dP_s]+\\ 
&\hspace{1cm}\frac{1}{2}\int_{k\delta t}^{(k+1)\delta t}\operatorname{trace}\left([dP_s]^T\nabla^2f[dP_s]\right)\\
\end{align*}

Condition $4)$ sets the last term to $o(\delta t)$ whereas condition $3)$ sets the first term to $\int_{k\delta t}^{(k+1)\delta t} (dP_s)^\mathcal{H}R_s^T+o(\delta t)$. Then, replacing $f$ by $\operatorname{Int}(P_{k\delta t},.)$
\begin{align*}
&\operatorname{Int}(P_{k\delta t},P_{(k+1)\delta t})-\operatorname{Int}(P_{k\delta t},P_{k\delta t})\\
&\hspace{1cm}=\int_{k\delta t}^{(k+1)\delta t}(dP_s)^{\mathcal{H}}R_s^T\\
\end{align*}
Now, thanks to condition $1)$, one gets that $\operatorname{Int}(P_{k\delta t},P_{k\delta t})=0$.

Separating the integral $\mathcal{S}_t$ into $t/\delta t$ short integrals gives, up to a remaining integral between $t$ and $kt/\delta t$, the following expression:
\begin{align*}
&\mathbb{E}\left[|\tilde{\mathcal{S}}_{n\delta t}-\mathcal{S}_t|^2\right]\\
&=\mathbb{E}\left[\left|\sum_{k\leq t/\delta t} \int_{k\delta t}^{(k+1)\delta t}<x_{k\delta t}-x_s,dz_s>\right|^2\right]\\
&=\mathbb{E}\left[\left|\sum_{k\leq t/\delta t} \int_{k\delta t}^{(k+1)\delta t}<x_{k\delta t}-x_s,(dR_s)R_s^T>\right|^2\right]\\
& \text{(as $\sigma_w^2I_n$ is orthogonal to $\mathfrak{so}(n)$)}\\
&\leq\sum_{k\leq t/\delta t}\mathbb{E}\left[ \left|\int_{k\delta t}^{(k+1)\delta t} <\hat{x}_{k\delta t}-\hat{x}_s,(dP_s)^\mathcal{H}R_s^T>\right|^2\right]
\end{align*}
The square variations of $x_t$ during a time $\delta t$ being bounded by $O(\delta t)$, the variation of $x_t$ are bounded by $O(\delta t)$ too. The integral is then bounded by $O(\delta t^2)$ thanks to the It\=o isometry property of the integral
\begin{align*}
&\mathbb{E}\left[|\tilde{\mathcal{S}}_{n\delta t}-\mathcal{S}_t|^2\right]\\
&=\mathbb{E}\left[\sum_{k\leq t/\delta t} \int_{k\delta t}^{(k+1)\delta t} O(\delta t)||(dP_s)^\mathcal{H}R_s^T||^2\right]\\
&=\sum_{k\leq t/\delta t}O(\delta t^2)\\
&=O(\delta t)
\end{align*}
This means that $S^\delta$ converges towards $\mathcal{I}$ in the mean square error sense when $\delta t$ shrinks to $0$ with a linear convergence rate.
\end{proof}

Based on this theorem, we can show that a linear interpolation between successive samples can be used to approximate the likelihood.

\begin{proposition}
\label{prop_interp_lin}
The interpolation function $\operatorname{Int}(P',P)=\frac{1}{2}d\Pi_{R'}\left(RR'^{T}-R'R^{T}\right)$ with $\Pi(R)=P$ and $\Pi(R')=P'$ satisfies the conditions of Theorem \ref{prop1}.
\begin{proof}
The conditions $1)$ and $2)$ from theorem 1 are directly verified. Using the notation $\sigma=v^{\mathcal H}R^T$, Condition $3)$ is verified as:
\begin{align*}
\nabla \operatorname{Int}(P,P)[v]&=\frac{d}{dt}\left.\operatorname{Int}(P,\exp(t\sigma)P)\right|_{t=0}\\
&=\frac{d}{dt}\frac{1}{2}\left(\exp(t\sigma)RR^T-R^T\exp(t\sigma)^T\right)|_{t=0}\\
&=\frac{d}{dt}\frac{1}{2}\left(\exp(t\sigma)+\exp(t\sigma)\right)|_{t=0}\\
&\text{ as $\sigma\in\mathfrak{so}(n)$}\\
&=\sigma=vR^T
\end{align*}

Finally, condition $4)$ is also satisfied because:
\begin{align*}
\nabla^2 \operatorname{Int}(P,P)[v]\\
&\hspace{-1cm}=\frac{d^2}{dtds}\operatorname{Int}(P,\exp(s\sigma)\exp(t\sigma)P)|_{t=0,s=0}\\
&\hspace{-1cm}=\frac{d}{dt}\frac{1}{2}\left(\exp(s\sigma)\exp(t\sigma)RR^T-\right.\\
&\hspace{-0.5cm}\left. R^T\exp(t\sigma)^T\exp(s\sigma)^T\right)|_{t=0}\\
&\hspace{-1cm}=\frac{d}{dt}\frac{1}{2}\left(\exp((t+s)\sigma)-\exp((t+s)\sigma)\right)|_{t=0}\\
&\hspace{-1cm}=0
\end{align*}
\end{proof}
\end{proposition}

Thanks to Proposition \ref{prop_interp_lin}, we are now able to implement a Monte-Carlo solution based on discrete samples. For a better readability, the interpolation term between two successive samples $\operatorname{Int}\left(P_{k\delta t},P_{(k+1)\delta t}\right)$ will be denoted $\delta P_k$ in the sequel.


\section{Practical solution to the filtering problem}
\label{sec_impl}
\subsection{Implementation via a Monte-Carlo method}

The particle filter is a Monte-Carlo method to approximate the solution given by Equation (\ref{eq_ks_st}) for non linear model. Despite that particle filters has been heavily used and studied \cite{Doucet}, the application of this method to perform estimation from partial observation on the Stiefel manifold has not been used before.

The main idea of the particle filter is to approximate the expectation in (\ref{eq_ks_st}) using the law of large numbers. Recall that:
$$
\rho_t(\phi) = \mathbb{E} \left[\phi(x')L_t(P,x')|\mathcal{P}_t \right]
$$
with $L_t$ the likelihood defined in Equation (\ref{eq_likelihood_st}). The process $x'_t$ is a copy of $x_t$ (in the sense with the same model of propagation) but, contrary to $x_t$, independent from $P_t$. 

Let $X^i_t$, with $i=1,\ldots,N$, be $N$ processes identical to $x_t$ called particles. They represent candidates to estimate the process $x_t$. The law of large numbers states that $\rho^N_t$ defined as:
$$
\rho^N_t(\phi)=\frac{1}{N}\sum_{i=1}^N \phi(X^i)L_t(P,X^i)
$$ 
will converge with $N$ almost surely to $\rho_t(\phi)$. Therefore, the previous equation gives an approximation of the solution by determining $\rho^N_t$ and normalizing it. It is noticeable that the particles are not observable and must be simulated. This means that the model of propagation (\ref{eq_model}) for $x_t$ should be known.

Furthermore, it is assumed that the process $P_t$ is not continuously observed and let denotes $\delta t$ the sampling time. These means that it is necessary to consider a time discretized version of $\rho^N_t$, denoted $\rho^N_n$, with $n=\lfloor t/\delta t\rfloor$, as:
$$
\rho^N_n(\phi)=\sum_i \phi(\tilde{X}^i)L_{\delta t}(\tilde{\mathcal{P}}_n,\tilde{X}^i_0,...,\tilde{X}^i_n)
$$
where $\tilde{X}^i_n=X^i_{n\delta t}$, $\tilde{\mathcal{P}}_n=\{P_{k\delta t},k\leq n\}$ and the likelihood $L_{\delta t}$ is defined by:
$$
L_{\delta t}(i,n)=\exp\left(\frac{1}{\sigma_w^2}\sum_{k\leq n}<H^i_k,\delta \tilde{P}_k>-\frac{1}{2}||H^i_k ||^2\right)
$$
with $H^i_k=\chi(X^i,\tilde{P}_k)$. It has been proved in Proposition \ref{prop_interp_lin} that the discrete likelihood $L_{\delta t}$ converges towards the likelihood $L_t$ from (\ref{eq_likelihood_st}).

In order to implement such a solution, two independent problems must be tackled:
\begin{itemize}
\item
The simulation of the particles $\tilde{X}^i$: \\
It will be supposed that $x_t$ is a Markov process. Consequently, $\tilde{X}^i$ is a Markov chain with transition kernel $q_{\delta t}$ defined after (\ref{eq_model}) and $\tilde{X}^i_{n+1}$ is directly sampled from $q_{\delta t}(\tilde{X}^{i}_{n},\ .\ )$.
\item
The computation of the likelihood $L_{\delta t}(\tilde{\mathcal{P}}_n,\tilde{X}^i_0,...,\tilde{X}^i_n)$:\\
Considering the last term of the Riemannian sums:
\begin{align*}
L_{\delta t}(i,n)=&\exp\left(\frac{1}{\sigma_w^2}\sum_{k\leq n}<H^i_k,\delta \tilde{P}_k>-\frac{1}{2}||H^i_k ||^2\right)\\
&=L_{\delta t}(i,n-1)l^i_{\delta t}
\end{align*}
where $l^i_{\delta t}=\exp\left(\frac{1}{\sigma_w^2}<H^i_n,\delta \tilde{P}_n>-\frac{1}{2}||H^i_n ||^2\right)$, this decomposition shows that the likelihood can be computed adaptively when new samples are available.
\end{itemize}

This leads to Algorithm \ref{alg_filt_part2}, here described to estimate the conditional distribution $\pi^{N,\delta t}_n(\phi)$ as
$$
\pi^{N,\delta t}_n(\phi)=\frac{\rho^{N,\delta t}_n(\phi)}{\rho^{N,\delta t}_n(1)}=\sum_{i=1}^N\phi(\tilde{X}^{i}_n)w^i_n
$$
where the coefficient
$$
w^i_n=\frac{L_{\delta t}(\tilde{\mathcal{P}}_n,\tilde{X}^i_0,...,\tilde{X}^i_n)}{\sum_{j=1}^N L_{\delta t}(\tilde{\mathcal{P}}_n,\tilde{X}^j_0,...,\tilde{X}^j_n)}
$$ is called the "weight" associated to the particle $i$.

\begin{algorithm}
\caption{Particle filter algorithm\label{alg_filt_part2}}
\begin{itemize}
\item
For the initialization, generate $N$ particles from a priori $p_0$: $\tilde{X}^i_0\sim p_0$ and set $w^i_0=1/N$.
\item
At a time $n> 0$:
\begin{enumerate}
\item
Propagate the particles $\tilde{X}^i_n \sim q_{\delta t}(\tilde{X}^i_{n-1},\ .)$  
\item
Update the weight $w^i_n$ of each particle as: $w^i_n=w^i_{n-1} l^i_{\delta t}$ with 
$$
l^i_{\delta t}= \exp\left(\frac{1}{\sigma_w^2}<H^i_n,\delta \tilde{P}_n>-\frac{1}{2}||H^i_n ||^2\right).
$$
\item
Normalize the weights: $w^i_n=w^i_n/\sum_j w^j_n$
\item
If $\left(\sum_i (w^i_n)^2\right)^{-1}<N/2$, generates $[m^1\ ...\ m^N]\sim\operatorname{multinomial}(w^1\ ...\ w^N)$ such that $\sum_i m^i=N$. Then, clones $\tilde{X}^i_n$ $m^i$-times and set $w^i_n=\frac{1}{N}$.
\item
Estimate $\pi_t(\phi)$ with $$\pi^{N,\delta t}_n(\phi)=\sum_{i}\phi(\tilde{X}^{i}_n) w^i_n.$$
\end{enumerate}
\end{itemize}
\end{algorithm}

The normalization step (step 3) is not only here to compute $\pi^{N,\delta t}_n$ instead of $\rho^{N,\delta t}_n$ but also to numerically stabilize the computation of the weights. As they are usually smaller than $1$, their consecutive multiplications lead to small values.

Step 4 is called resampling. It is here to prevent a degeneracy due to the finite number of particles. Indeed, the particles are propagating without any restriction or drift imposed by the observation. Without resampling, particles would just explore the space and as they tend to drift away from $x_t$, they would become a bad approximation of $x_t$, because the mean square error $\mathbb{E}[(X^i_t-x_t)^2]$ is linearly growing with time. The number of particles being fixed, their weights quickly degenerate as they are diffusing  away from $x_t$. Due to the normalization step, this leads to the concentration of all the ponderation into one single particle. Even if this particle is the best candidate amongst the all the particles, the mean square error is still linearly growing.

The resampling step consists in killing the particles far away from $x_t$ (in fact, killing the particles with low weights) and cloning the remaining ones. In order to measure if the particles are scattered away from $x_t$, one commonly used criteria is a threshold based on the Effective Sample Size ($ESS_w$) defined as
$$
ESS_w=\left(\sum_i (w^i_n)^2\right)^{-1}.
$$
When $ESS_w$ is lower that say $N/2$, then particles need to be resampled (the criteria $ESS_w$ is small when only a few particles have a preponderant weight). To resample the particles, one can for example sample $[m^1\ ...\ m^N]$ from a multinomial distribution
$$
[m^1\ ...\ m^N]\sim\operatorname{multinomial}(w^1\ ...\ w^N)
$$
such that $\sum_i m^i=N$ (to keep the number of particles constant) and clone the $i^{\text{th}}$ particle $m^i$ times. If $w^i$ is high (particle with a good likelihood, thus a good candidate), then $m^i$ should be high too. This effect will tend to keep only the good candidates, based on the likelihood. However, instead of resampling when the Effective Sample Size becomes too low, resampling is made after a fixed given time. In fact, resampling can be realized at every iteration but it is, computationally speaking, expensive and does not bring noticeable improvements \cite{Doucet}.

\subsection{Simulation results}
This subsection describes the results obtained from a numerical implementation of the particle filter detailled in Algorithm \ref{alg_filt_part2}. For this simulation, the chosen Stiefel manifold was the sphere $V_{3,1}=S^2$. The process $x_t\in \mathbb{R}^3$ is a Brownian motion with a unit variance. The variance of the noise is fixed to $\sigma_w^2=1$.

\begin{figure}
\centering
\subfigure{\epsfig{figure=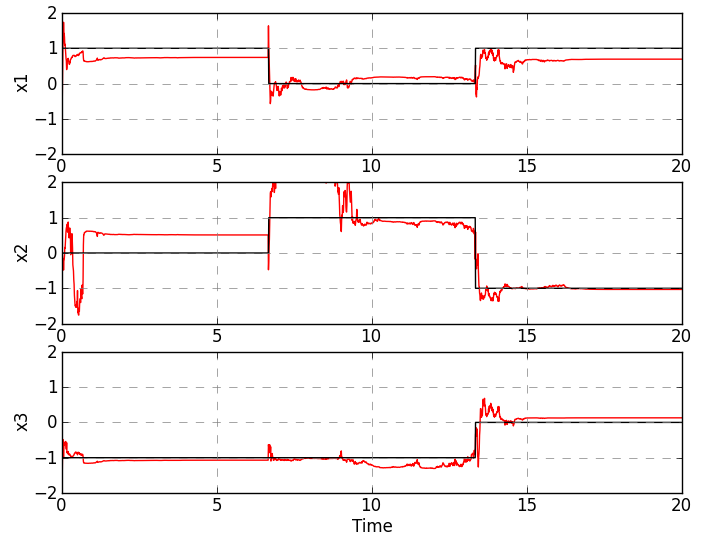,width=0.4\textwidth}}\\ 
\subfigure{\epsfig{figure=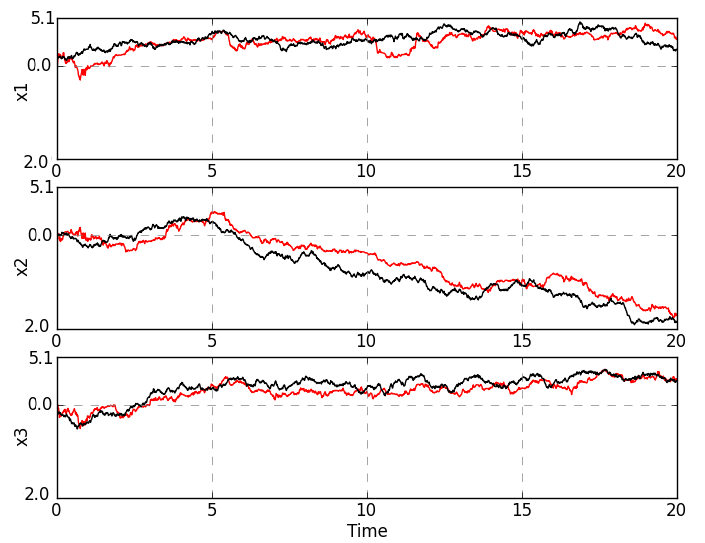,width=0.4\textwidth}}\\
\subfigure{\epsfig{figure=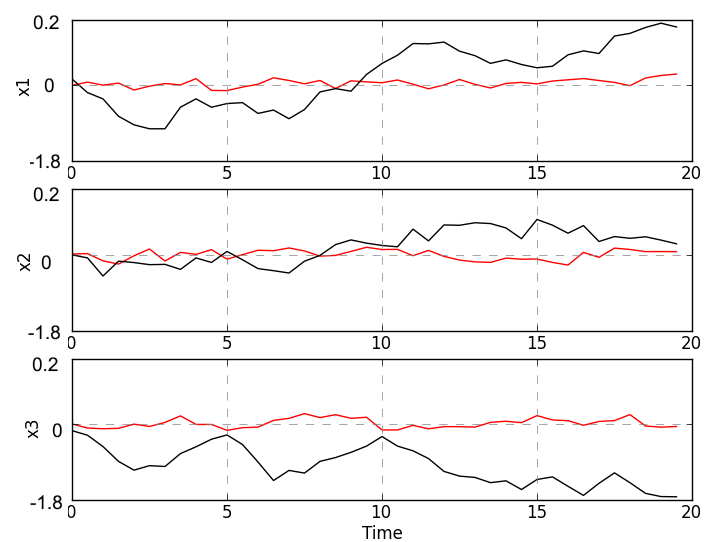,width=0.4\textwidth}}\\
\caption[Evolution of the estimation (red) of each component of $x_t$]{\label{fig_stair} Evolution of the estimation (red) of each component of $x_t$ (black), namely $x_1$, $x_2$ and $x_3$. (Top) The model used for $x_t$ is a stair function and the times at which changes occur are known. The sampling time is $\delta t=10^{-2}$. The particle algorithm properly converges. (Middle) $x_t$ is a Brownian motion in $\mathbb{R}^3$ with unit variance. The algorithm can still estimate properly as the variations are not too fast. (Bottom) For the same model, the sampling time is reduced to $\delta t=0.5$ and the filter is not able to track $x_t$ anymore. The component in the vertical space changes too quickly.}
\end{figure}

To approximate $x_t$, $N=500$ particles are generated from a normal prior distribution $p_0$ centered around the origin with a variance of $2$. Using more particles does not significantly improve the results. Considering a bad prior for generating the particles is not a big issue as the resampling step quickly eliminates the wrong candidates for the estimation. In the first two cases displayed in Figure \ref{fig_stair} (Top and Bottom), the time step for the observation is $\delta t=10^{-2}s$ and it is set to $\delta_t=0.5s$ for the last case (Bottom). The time step for creating the simulation has been fixed to $10^{-3}s$, which is sufficient to consider the process continuous with respect to the observation time step. Figure \ref{fig_stair} illustrates the results obtained for the estimation when the state $x_t$ is a stair function (Top), then when $x_t$ varies slowly (Middle) and finally when the sampling period is too large to be able to track properly the evolution of $x_t$ (Bottom). As long as $x_t$ is slowly varying with respect to $\delta t$, the algorithm is able to completely estimate $x_t$. In the case where $x_t$ follows a stair function model, one could use a classical algorithm to detect abrupt changes in $x_t$ in order to estimate the time instants where the particles should be sampled \cite{Basseville1993} (this was not effectively implemented in the results presented in the Top figure of \ref{fig_stair}, where it was simply assumed that the time where changes occur were known). When a change is detected, the particles are once again sampled from the initial priori to converge toward the new value. In the presented case (Top of figure \ref{fig_stair}), the particles will not drift away because they are at a constant position (they propagate with the same model as $x_t$). Particles strongly merge when they are resampled, leaving less and less possible choice. 

It was mentioned earlier that the vertical component could not be estimated. However, the vertical space is defined with respect to the observation point $P_t$. As $P_t$ will evolve on $V_{n,k}$, the vertical space will change too and the component on the initial vertical space can thus be estimated. As a consequence, it is finally possible to completely estimate the angular velocity. Note that in the case where $x_t$ evolves slowly, the vertical component can still be estimated.

\begin{figure}
\centering
\epsfig{figure=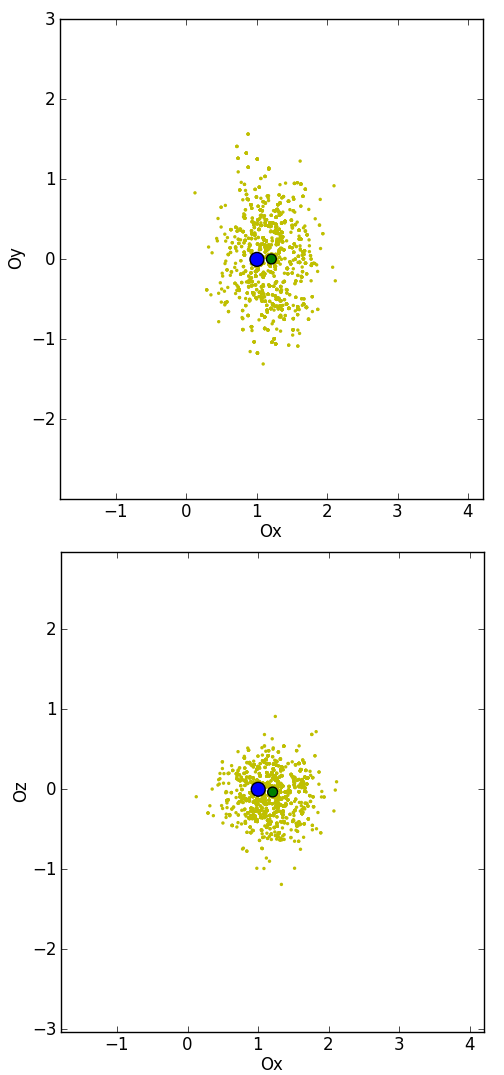,width=6cm,height=11cm}
\caption[Projection of the particles, the state to estimate and the estimation on the plane $Ox,Oy$ and $Ox,Oz$.]{Projection of the particles (yellow), the state to estimate (blue) and the estimation (green) on the plane $Ox,Oy$ and $Ox,Oz$. At the time of the snapshot, $P_t\approx [0,1,0]^T$. Consequently, the component along $Oy$ cannot be estimated. The particles are distributed within an ellipse whose large axis is directed along $Oy$.
\label{fig_ellipse}}
\end{figure}

Now, if the angular velocity is evolving too quickly (with respect to the amount of time particles need to converge), it will not be possible to estimate the vertical component of $x_t$ from the observation of $P_t$. The particles are distributed within an ellipse (see Figure \ref{fig_ellipse}, Top), whose large axis coincides with the direction of $P_t$. Along the direction of $P_t$, the estimation of $x_t$, the empirical average of the particles (yellow dots), is not satisfactory. However, in the other direction, the estimation is correct. This is due to the fact that the particles can only track the component of $x_t$ that has an impact on $P_t$. As the innovation term in Equation (\ref{eq_likelihood_st}) is $<H_s,dP_s>$, only the horizontal component can be observed and therefore estimated. Consequently, particles propagating along the vertical subspace (wich is, for $S^2$, the line defined by $P_t$) are not penalized (their weight does not decrease) and are still considered as good candidates. As a consequence, the estimate is correct in the horizontal direction, but not in the vertical one.
This last comment can be understood as highlighting the cases where the proposed particle filter failed at estimating correctly the complete set of components of $x_t$ due to the lack of information in the observation.   

In the next section, we consider a special case of our problem, namely when observations are complete.

\section{Special case for observation from $SO(n)$}
\label{sec_son}
\subsection{Optimal filtering in $SO(n)$}
The special case when observations are in $V_{n,n}$ can be described using the technique presented in previous sections. However, the fact that $P_t$ is in one of the continuous component of $O(n)$ allows us to treat it also in a different manner. As it is described in Section \ref{sec_geom}, the process $P_t$ can be considered like a process with values in $SO(n)$ (which is one of the two continuous components of $O(n)$) without any loss of generality.

As $P_t\in SO(n)$, then the map $\Pi$ is the identity map and the application $\chi$ is invertible. Therefore, there is no vertical space over $P_t$ and the tangent space is simply the horizontal space. Equation (\ref{eq_antidev_st2}) then reads:
\begin{equation}
dz_t=x_tdt+\circ dw_t
\label{eq_antidev_son}
\end{equation}
where $dz_t=(\circ dP_t)P_t^T$. In this case, a numerical method is no more required as the full process is observed, \textit{i.e} $H_t=x_t$. As the increments in Equation (\ref{eq_antidev_son}) are independent, and the noise is additive, a classic Kalman filter can be used where the anti-development $z_t$ replaces the observed process $P_t$. The conditional distribution of $x_t$ is then a Gaussian distribution with mean $\mu_t$ and a variance $V_t$ such that:

\begin{equation}
\begin{array}{rcl}
d\mu_t&=&F\mu_t dt+\frac{1}{\sigma_w^2}V_t(dz_t-\mu_t dt)\\
dV_t&=&FV_t+V_tF^T-\frac{1}{\sigma_w^2}V_t^2+\sigma_b^2.\\
\label{eq_kb_son}
\end{array}
\end{equation}
This solution represents the optimal filter for observation in $SO(n)$. However, in practice, the same issue as in Section \ref{sec_filt_st} occurs due to the discretization of the observation.

\subsection{Implementation of the solution}
The discrete nature of the observation of $P_t$ does not allow to continuously determine $z_t$. An approximation must be performed using an interpolation function. The linear interpolation function $\operatorname{Int}(P',P)=\frac{1}{2}d\Pi_{R'}\left(PP'^{T}-P'P^{T}\right)$ still converges towards the continuous solution as the sampling period $\delta t$ shrinks to $0$. Amongst all the possible functions, one can also choose to use $\operatorname{Int}(P',P)=\log(PP'^{T})$ as it satifies the conditions of Theorem \ref{prop1}.
\begin{proposition}
\label{prop_interp_geod}
The interpolation function $\operatorname{Int}(P',P)=\log(PP'^{T})$ satisfies the conditions of Theorem \ref{prop1}.

\begin{proof}
The conditions $1)$ and $2)$ are direct. For condition $3)$, an element $v\in T_PSO(n)$ is described with right invariant vector fields $V=\sigma P$ with $\sigma\in\mathfrak{so}(n)$.
\begin{align*}
&\nabla \operatorname{Int}(P,P)[v]\\
&=\frac{d}{dt}\operatorname{Int}(P,\exp(t\sigma)R)|_{t=0}\\
&=\frac{d}{dt}\log\left(\exp(t\sigma)P^T\right)|_{t=0}\\
&=\frac{d}{dt}t\sigma|_{t=0}\\
&=\sigma=vP^T
\end{align*}
Condition $4)$ is also satisfied as:
\begin{align*}
&\nabla^2 \operatorname{Int}(P,P)[v]\\
&=\frac{d^2}{dtds}\operatorname{Int}(P,\exp(s\sigma)\exp(t\sigma)P)|_{t=0,s=0}\\
&=\frac{d^2}{dtds}\log\left(\exp((t+s)\sigma)P^T\right)|_{t=0}\\
&=\frac{d^2}{dtds}(t+s)\sigma|_{t=0,s=0}\\
&=0
\end{align*}
\end{proof}
\end{proposition}

The function $\operatorname{Int}(P',P)=\log(PP'^{T})$ is called geodesic interpolation. This interpolation function has a higher computational cost than the linear interpolation one, but it is also more accurate for approximating the incremental term $dz_t$.

In order to illustrate this point, Figures \ref{fig_comp_interp} and \ref{fig_comp_interp2} present the difference between a Kalman filter directly applied on the discrete observation of $z_t$ (which is an optimal filter), and the different methods of interpolation based on the discrete observation of $P_t$.

Because the equation of the variance $V_t$ does not depend on the observation and is just an isolated differential equation (with respect to the innovation), the variance for each algorithm used in Figure \ref{fig_comp_interp} is the same (considering that the initial value $V_0$ has always been chosen with the same value) and is consequently not displayed for this kind of comparison.

\begin{figure}
\centering
\centerline{\epsfig{figure=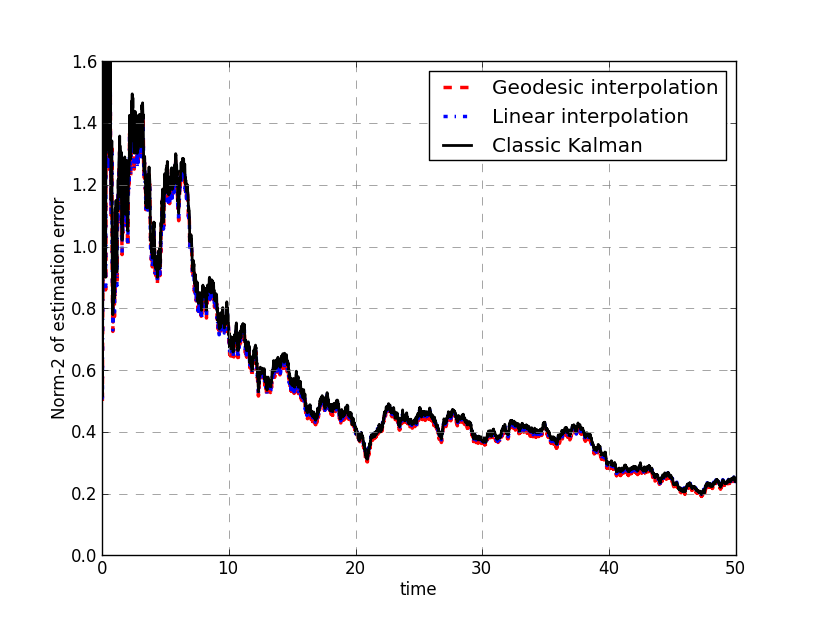,width=9cm,height=7cm}}
\caption[Comparison between linear interpolation, geodesic interpolation, and classic Kalman filter with $\delta t=0.01s$]{Comparison of estimation results (MSE) between linear interpolation, geodesic interpolation, and classic Kalman filter with $\delta t=0.01s$. The state to estimate is constant and the variance of the observation is $\sigma_w^2=1$
\label{fig_comp_interp}}
\end{figure}

In the presented simulations, the process $x_t$ is considered constant. The use of another model might not imply any significant change. The process $z_t$ is first generated with a small time step ($10^{-4}s$) and with a variance $\sigma_w^2=1$ \emph{via} an Euler scheme from the stochastic differential equation (\ref{eq_antidev_son}).

At the same time, a rotational process $P_t$ is constructed from $z_t$. The construction of $z_t$ and $P_t$ is then realized with a time step small enough to consider them as "time continuous".

The process $P_t$ is then sampled with a time step $\delta t$ with $\delta t>>10^{-4}$. The performances of the filter from (\ref{eq_kb_son}) with different methods of interpolation (linear and the geodesic) to approximate $\delta z_{k\delta t}$ are displayed in figure \ref{fig_comp_interp}. In parallel, as $z_t$ has been continuously generated, these performances are compared to a classical Kalman filter taking directly $\delta z_{k\delta t}=z_{(k+1)\delta t}-z_{k\delta t}$. Recall that this term is not available due to the discrete observation of $P_t$. Knowing that in the case of additive noise, the Kalman filter is optimal \cite{Jazw} and because the anti-development is in one-to-one correspondance with the observation, the Kalman filter is used here as a reference filter (as it is known that it is not possible, in the mean square sense, to outperform it).

\begin{figure}
\centering
\centerline{\epsfig{figure=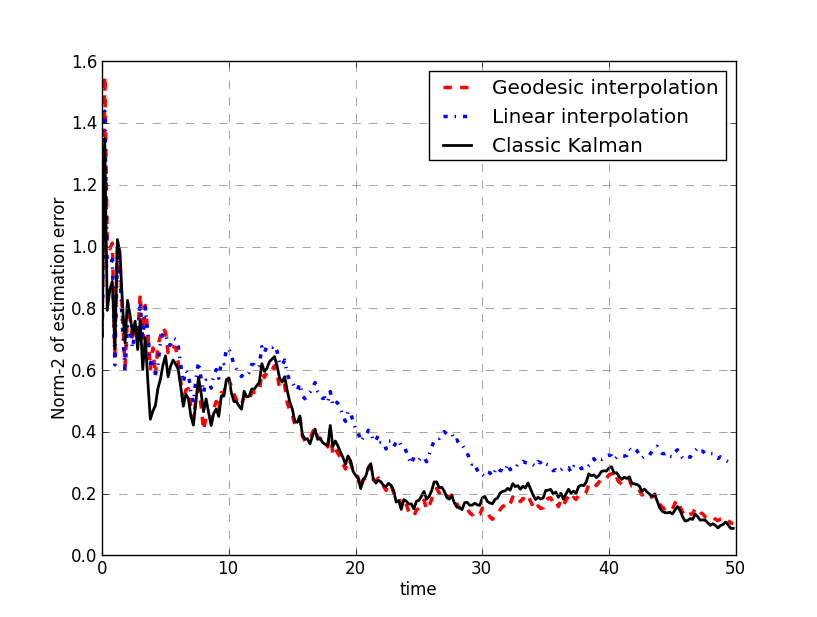,width=9cm,height=7cm}}
\caption[Comparison of estimation results (MSE) between linear interpolation, geodesic interpolation, and classic Kalman filter with $\delta t=0.2s$]{Comparison of estimation results (MSE) between linear interpolation, geodesic interpolation, and classic Kalman filter with $\delta t=0.2s$. The state to estimate is constant and the variance of the observation is $\sigma_w^2=1$.\label{fig_comp_interp2}}
\end{figure}

The results presented in Figure \ref{fig_comp_interp} have been obtained with $\delta t=10^{-2}s$. This time step is small enough so that similar performances for the different methods of interpolation (geodesic and linear) as they should converge towards the same solution. However, differences start to appear if $\delta t$ has larger values. 

In Figure \ref{fig_comp_interp2}, $\delta t$ is increased to $0.2s$. The variation in the observation time step changes the approximation of the process $z_t$ and differences between the two interpolation methods (linear and geodesic) are visible. The linear interpolation is not as accurate as the geodesic interpolation. The linear approximation adds another source of error to the estimation, as a drawback to its simpler computational form.

\begin{figure}
\centering
\centerline{\epsfig{figure=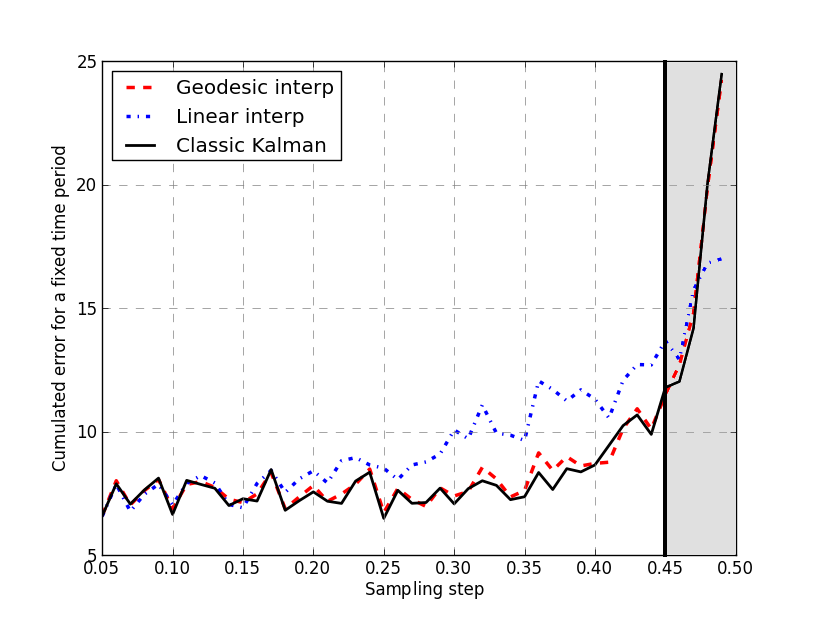,width=9cm,height=7cm}}
\caption[Cumulated error of estimation for $0\leq t \leq 50$ for the linear interpolation, the geodesic interpolation, and the classic Kalman filter for different values of $\delta t$]{Cumulated error of estimation for $0\leq t \leq 50$ for the linear interpolation, the geodesic interpolation, and the classic Kalman filter for different values of $\delta t$. The state to estimate is constant and the variance of the observation is $\sigma_w^2=1$. Each point represents the average of 20 simulations configured with the same parameters. For a sampling time larger than $0.45$ (arbitrarily chosen), the performances are too poor to consider the filters to converge properly any-more. 
\label{fig_comp_step}}
\end{figure}

Finally, in order to observe the influence of the sampling step $\delta t$ on the performance of each interpolation method, Figure \ref{fig_comp_step} illustrates the evolution of the cumulated error for $0\leq t\leq 50$ at a fixed sample step for each method. Just like previously, the cumulated errors should be compared with the cumulated error induced by a proper Kalman filter. It appears that the geodesic interpolation does not create another error term despite that it is an approximation of $\delta z_{n\delta t}$. The cumulated error is the same as the Kalman filter one. The linear interpolation, however, is adding a supplementary error term. As the sampling step is increasing, the approximation is worse and worse. In the end, for $\delta t>0.45s$, the sampling step is too high for any filter to perform a proper estimation of the state. For such cases, solutions such as the extended Kalman filter presented in \cite{Bourmaud} should be privileged.

\section{Conclusion}
In this paper, a solution to the problem of filtering from partial observations is presented. The observed process takes its values in the Stiefel manifold while the signal of interest evolves on the rotation group. Due to the multiplicative nature of the noise, standard methods cannot be applied directly. A solution relying on the construction of an intermediate process, namely the anti-development, is proposed. This solution uses a Monte-Carlo method to overcome the problem of missing information. It is shown that the proposed algorithm allows, in certain contexts, to recover the whole set of components of the signal of interest despite the lack of observations. Finally, in the special case where the entire process can be observed, an optimal filter, together with interpolation methods, is given. This filter can be interpreted as a Kalman filter for observations on the rotation group.

\bibliographystyle{IEEEtran}
\bibliography{biblio}

\begin{thebibliography}{10}
\providecommand{\url}[1]{#1}
\csname url@samestyle\endcsname
\providecommand{\newblock}{\relax}
\providecommand{\bibinfo}[2]{#2}
\providecommand{\BIBentrySTDinterwordspacing}{\spaceskip=0pt\relax}
\providecommand{\BIBentryALTinterwordstretchfactor}{4}
\providecommand{\BIBentryALTinterwordspacing}{\spaceskip=\fontdimen2\font plus
\BIBentryALTinterwordstretchfactor\fontdimen3\font minus
  \fontdimen4\font\relax}
\providecommand{\BIBforeignlanguage}[2]{{%
\expandafter\ifx\csname l@#1\endcsname\relax
\typeout{** WARNING: IEEEtran.bst: No hyphenation pattern has been}%
\typeout{** loaded for the language `#1'. Using the pattern for}%
\typeout{** the default language instead.}%
\else
\language=\csname l@#1\endcsname
\fi
#2}}
\providecommand{\BIBdecl}{\relax}
\BIBdecl

\bibitem{Bloch}
A.~M. Bloch and J.~E. Marsden, ``Stabilization of rigid body dynamics by the
  energy-casimir method,'' \emph{Systems and Control Letters}, vol.~14, no.~4,
  pp. 341 -- 346, 1990.

\bibitem{Lovera}
A.~Astolfi and M.~Lovera, ``Global spacecraft attitude control using magnetic
  actuators,'' in \emph{American Control Conference, 2002. Proceedings of the
  2002}, vol.~2, 2002, pp. 1331--1335 vol.2.

\bibitem{Zamani11}
M.~Zamani, J.~Trumpf, and R.~Mahony, ``Near-optimal deterministic filtering on
  the rotation group,'' \emph{Automatic Control, IEEE Transactions on},
  vol.~56, no.~6, pp. 1411--1414, June 2011.

\bibitem{barrau13}
A.~Barrau and S.~Bonnabel, ``Intrinsic filtering on so(3) with discrete-time
  observations,'' in \emph{Decision and Control (CDC), 2013 IEEE 52nd Annual
  Conference on}, Dec 2013, pp. 3255--3260.

\bibitem{zamani13}
M.~Zamani, J.~Trumpf, and R.~Mahony, ``Minimum-energy filtering for attitude
  estimation,'' \emph{Automatic Control, IEEE Transactions on}, vol.~58,
  no.~11, pp. 2917--2921, Nov 2013.

\bibitem{Landis14}


\bibitem{Bourmaud}
G.~Bourmaud, R.~M\'egret, A.~Giremus, and Y.~Berthoumieu, ``Discrete extended
  kalman filter on lie groups,'' \emph{European Signal Processing Conference},
  2013.

\bibitem{Canet}
P.~Canet, ``Kalman filter estimation of angular velocity and acceleration
  on-line implementation,'' Report, 1994.

\bibitem{Edelman}
A.~Edelman, T.~A. Arias, and S.~T. Smith, ``The geometry of algorithms with
  orthogonality constraints,'' \emph{Siam J. Matrix Anal. Appl}, vol.~20,
  no.~2, pp. 303--353, 1998.

\bibitem{Micka}
O.~Micka and A.~J. Weiss, ``Estimating frequencies of exponentials in noise
  using joint diagonalization,'' \emph{Signal Processing, IEEE Transactions
  on}, vol.~47, no.~2, pp. 341--348, 1999.

\bibitem{chikuze}
Y.~Chikuse, \emph{Statistics on special manifolds}.\hskip 1em plus 0.5em minus
  0.4em\relax Springer Lecture notes in Statistics, 2003.

\bibitem{absil}
P.-A. Absi, R.~Mahony, and R.~Sepulchre, \emph{Optimisation algorithms on
  matrix manifolds}.\hskip 1em plus 0.5em minus 0.4em\relax Princeton
  University Press, 2008.

\bibitem{Ozdemir}
M.~K. Ozdemir and H.~Arslan, ``Channel estimation for wireless ofdm systems,''
  \emph{Communications Surveys Tutorials, IEEE}, vol.~9, no.~2, pp. 18--48,
  2007.

\bibitem{Lui09}
Y.~M. Lui, J.~R. Beveridge, and M.~Kirby, ``Canonical stiefel quotient and its
  application to generic face recognition in illumination spaces,'' in
  \emph{Proceedings of the 3rd IEEE international conference on Biometrics:
  Theory, applications and systems}, ser. BTAS'09.\hskip 1em plus 0.5em minus
  0.4em\relax IEEE Press, 2009, pp. 431--438.

\bibitem{Abs08}
P.-A. Absil, R.~Mahony, and R.~Sepulchre, \emph{Optimization Algorithms on
  Matrix Manifolds}.\hskip 1em plus 0.5em minus 0.4em\relax Princeton, NJ:
  Princeton University Press, 2008.

\bibitem{Jazw}
A.~H. Jazwinski, \emph{Stochastic Processes and Filtering Theory}.\hskip 1em
  plus 0.5em minus 0.4em\relax Academic press, 1970.

\bibitem{Doucet}
A.~Doucet and A.~M. Johansen, ``A tutorial on particle filtering and smoothing:
  fifteen years later,'' 2011.

\bibitem{Basseville1993}
M.~Basseville and I.~Nikiforov, \emph{Detection of abrupt changes: Theory and
  application}, ser. Prentice hall information and system sciences
  series.\hskip 1em plus 0.5em minus 0.4em\relax Prentice Hall, 1993.

\end{thebibliography}

\end{document}